\documentclass[reqno, centertags, 12pt, draft]{article}

\usepackage{amscd}

\usepackage{mathrsfs}

\usepackage{amsmath,amssymb, amsfonts, amsthm, latexsym,titlesec,tikz}

\newtheorem{lem}{Lemma}[section]

\newtheorem{cor}[lem]{Corollary}
\newtheorem{propos}[lem]{Proposition}
\newtheorem{thm}[lem]{Theorem}

\newtheorem{defin}[lem]{Definition}

\newcommand{\comment}[1]{}

\newcommand{\abeta}{\gamma}
\newcommand{\abs}[1]{\left|#1\right|}
\newcommand{\boxL}{\Lambda_L^d}
\newcommand{\boxLp}{\Lambda_L^{d+1}}
\newcommand{\limL}{\underset{L\rightarrow\infty}\lim}

\newcommand{\path}{p}
\newcommand{\poly}{f}
\newcommand{\polyb}{g}
\newcommand{\Poly}{F}

\newcommand{\supp}{\textnormal{supp}}

\newcommand{\tildq}{\widetilde \polyb}
\newcommand{\N}{\mathbb{N}}

\newcommand{\Z}{\mathbb{Z}}
\newcommand{\Zd}{\mathbb Z^d}
\newcommand{\R}{\mathbb{R}}
\newcommand{\Ex}[1]{\mathbb E\left[#1\right]}
\newcommand{\Par}[1]{\left(#1\right)}
\newcommand{\SPar}[1]{\left[#1\right]}
\newcommand{\Set}[1]{\left\{#1\right\}}
\newcommand{\SetPred}[2]{\left\{#1\ \middle|\ #2\right\}}
\newcommand{\Tr}[1]{\textnormal{Tr}\left(#1\right)}
\newcommand{\Var}[1]{\textnormal{Var}\left(#1\right)}
\newcommand{\sVar}[1]{\textnormal{Var}(#1)}
\newcommand{\Cov}[1]{\textnormal{Cov}\left(#1\right)}

\numberwithin{equation}{section}

\begin{document}

\title{Spectral fluctuations for the multi-dimensional Anderson model}
\author{Yoel Grinshpon, Moshe J.\ White \thanks{Institute of Mathematics, The Hebrew University of Jerusalem, Jerusalem, 91904, Israel.
Emails: yoel.grinshpon@mail.huji.ac.il, moshe.white@mail.huji.ac.il}}

\maketitle

\begin{abstract}
In this paper, we examine fluctuations of polynomial linear statistics for the Anderson model on $\Z^d$ for any potential with finite moments. We prove that if normalized by the square root of the size of the truncated operator, these fluctuations converge to a Gaussian limit. For a vast majority of potentials and polynomials, we show that the variance of the limiting distribution is strictly positive, and we classify in full the rare cases in which this does not happen.
\end{abstract}

\section{Introduction}
The purpose of this paper is to study fluctuations of the eigenvalue counting measure for the Anderson model on $\Z^d$. We denote $\abs{n}=\sum_{v=1}^d n_v$ for any $n\in\Zd$, and write $n\sim m$ for $n,m\in\Zd$ if and only if $\abs{n-m}=1$. Define the operator $H: \ell^2(\Z^d)\longrightarrow \ell^2(\Z^d)$ by
$$(Hu)_n =\Par{\Delta u}_n +\Par{X u}_n = \sum_{m \sim n} u_m+X_n \cdot u_n$$
where $\Set{X_n}_{n\in\Zd}$ is an array of independent, identically distributed (iid) random variables with finite moments, satisfying $\Ex{X_n}=0$. We denote the distribution of each variable $X_n$ by $\textnormal{d}\rho$, which will henceforth be referred to as \emph{the underlying distribution}.

In this paper, we aim to study the fluctuations of the counting measure for the eigenvalues of finite volume approximations. Explicitly, we study fluctuations of polynomial linear statistics of finite volume truncations of $H$: for any $L\in\N$, denote
$$\Lambda_L=[-L,L]\cap\Z,$$
and let $H_L$ be the truncation of $H$ to the cube $\Lambda_L^d \subset \Zd$. That is,\\
 $H_L=1_{\Lambda_L^d} H 1_{\Lambda_L^d}$, where
$$\Par{1_{\Lambda_L^d}(u)}_n=\begin{cases}u_n & n\in\Lambda_L^d\\
0 & n\notin\Lambda_L^d\end{cases}.$$

We denote by $N\left(0,\sigma^2\right)$ the normal distribution on $\R$ with mean $0$ and variance $\sigma^2$, and denote by $\overset{d}{\longrightarrow}$ convergence in distribution. We agree that the zero random variable is also normal, by allowing $\sigma^2=0$ (in this case we say the distribution is \emph{degenerate}).\\

The \emph{empirical measure} of $H_L$ is the measure
$$\textnormal{d}\nu_L=\frac{1}{(2L+1)^{d/2}}\sum_{i=0}^{\abs{\Lambda_L^d}} \delta_{\lambda_i^{(\Lambda_L^d)}}$$
where $\left \{\lambda_{1}^{(\Lambda_L^d)},\lambda_2^{(\Lambda_L^d)},\ldots,\lambda_{\abs{\Lambda_L^d}}^{(\Lambda_L^d)} \right \}=\sigma\Par{H_L}$ are the eigenvalues of $H_L$ (counting multiplicity), and $\delta_\lambda$ is the Dirac measure at $\lambda$. When the empirical measure has a limit as $L \rightarrow \infty$, this limit is known as the density of states of $H$. In our case, it is known that the random measure $\textnormal{d}\nu_L$ converges weakly almost surely to a deterministic measure $\textnormal{d}\nu$ (see e.g.\ \cite{AW} and references within).

We want to focus on asymptotics of the fluctuations of $\textnormal{d}\nu_L$. A natural way to study this is using \emph{linear statistics} for polynomials, i.e., random variables of the form $\int\poly\textnormal{d}\nu_L=\frac{1}{\Par{2L+1}^{d/2}}\Tr{\poly(H_L)}$ for some polynomial $\poly(x)\in\R[x]$.

Fluctuations of the truncated eigenvalues $\lambda_i^{\Par{\Lambda_L^d}}$ are assumed to be associated to continuity properties of the spectral measures. There are several results indicating this is indeed true. Minami \cite{Minami} studied the microscopic scale of the eigenvalues of the Anderson model in $\Z^d$, after Molchanov \cite{M} did the same for the continuous case in one dimension.
Minami proved that under certain conditions that ensure localization with exponentially decaying eigenfunctions, the eigenvalues of the Anderson model have Poisson behavior on the microscopic scale. For $d=1$, it is well known that localization holds for any ergodic non-deterministic potential \cite{AW}. However, for $d\ge 3$ and for sufficiently low energies, it is conjectured that $H$ has extended states, i.e., the spectrum of $H$ has an absolutely continuous component.

We now state our main theorem:
\begin{thm} \label{main_thm}
Let $\poly(x) \in\R[x]$ be a non-constant polynomial. Then
$$\frac{\Tr{\poly\Par{H_L}}-\Ex{\Tr{\poly\Par{H_L}}}}{(2L+1)^{d/2}} \overset{d}{\longrightarrow} N(0,\sigma(\poly)^2)$$
as $L\rightarrow\infty$, where:
\begin{enumerate}
\item If the underlying distribution $(\textnormal{d}\rho)$ is supported by more than three points, then $\sigma(\poly)^2>0$.
\item If the underlying distribution is supported by exactly two points, there exist polynomials $\polyb_2,\polyb_3,\polyb_5\in\R[x]$, of degrees $2,3,5$ respectively, such that $\sigma(\poly)^2=0$ if and only if $\poly\in\textnormal{span}_\R\Set{\polyb_5,\polyb_3,\polyb_2,1}$.
\item If the underlying distribution is supported by exactly three points, there exists a polynomial $\tildq_3\in\R[x]$ of degree $3$, such that $\sigma(\poly)^2=0$ if and only if $\poly\in\textnormal{span}_\R\Set{\tildq_3,1}$.
\end{enumerate}
\end{thm}
The polynomials $\polyb_2,\polyb_3,\polyb_5, \tildq_3$ depend on $\textnormal{d}\rho$ as well as on the dimension $d$, and are given explicitly in Propositions \ref{deg_23_var} and \ref{deg_5_var} below.

The study of fluctuations of finite truncations of the Anderson Model has received a considerable amount of attention, although most results focus on the one-dimensional case. Reznikova \cite{Rez} proved a CLT for the eigenvalue counting function of the truncated Anderson model in 1-dimension. Kirsch and Pastur \cite{KP} proved a CLT for the trace of truncations of the Green function of the Anderson model in one dimension. Recently, Pastur and Shcherbina \cite{PS} extended this result to other functions of $H$.

In our proof we shall compute the trace of powers of $H_L$ by counting paths on the associated lattice. Path counting and weighted path counting is commonly used in the study of random Schrodinger operators and in the study of random matrices (see, e.g., \cite{AW} and \cite{RMT} and references therein).

This paper can be viewed as a second paper in a series, continuing the work of Breuer with the authors \cite{BGW}. In the previous paper, path counting was used to prove a central limit theorem (CLT) for a decaying model over $\N$. In this paper, the methods have been modified to apply to the Anderson model over $\Zd$ for general $d\in\N$. Each of the papers is self contained, but there are many parallels in the overall structure of the paper and propositions.

The rest of the paper is organized as follows: In Section 2 we set up our definitions, and prove that `typical' diagonal elements in the matrix representation of $H_L^k$ have a combinatorial description (using path counting). In Sections 3 and 4 we prove our main theorem - in Section 3 we show that fluctuations of $\Tr{f\Par{H_L}}$ converge to a normally distributed random variable, and in Section 4 we classify all cases in which the limit distribution is non-degenerate. The final proof of Theorem \ref{main_thm} appears at the end of Section 4. We conclude with Section 5 (which is independent from the rest of the paper) in which we state and prove a CLT for $m$-dependent random variables indexed by $\Zd$, which implies the CLT we use in Section 3.

\textbf{Acknowledgments.} We are deeply grateful to Jonathan Breuer for his generous guidance and support throughout this research.

Research by YG was supported by the Israel Science Foundation (Grant No. 399/16). Research by MW was supported in part by the Israel Science Foundation (Grant No. 1612/17) and in part by the ERC Advanced Grant (Grant No. 834735).

\section{Definitions and preliminaries}
Fix $d\in\N$. As stated in the Introduction, we explore the random operator $H:\ell^2\Par{\Z^d}\to \ell^2\Par{\Z^d}$. 
It is useful to decompose $H$ as
\begin{equation} \label{operator_decomp}
H=V+\sum_{v=1}^d U_v + \sum_{v=1}^d D_v,
\end{equation}
where $V$ is the random potential operator, and each $U_v$ (respectively $D_v$) is the operator shifting forward (respectively backward) in direction $v$. In other words, let $e_1,e_2,\ldots,e_d$ denote the standard generators of $\Zd$ as a free abelian group. Then for every $n\in\Zd$ and $u\in\ell^2(\Zd)$ we have $(Vu)_n=X_n u_n$, and for every $1\leq v\leq d$ we have $(U_v u)_n=u_{n+e_v}$ and $(D_v u)_n=u_{n-e_v}$. A corresponding decomposition is also given for every finite volume truncation, $H_L$.


Our theorem deals with the asymptotic behavior (as $L\rightarrow\infty$) of\\
 $\Tr{\poly\Par{H_L}}$, for polynomials $\poly\in\R[x]$. We consider $\Tr{\poly\Par{H_L}}$ as a polynomial in the variables $\SetPred{X_n}{n\in\Zd}$. To slightly ease notation, we denote our variables by a lowercase Latin letter (such as $x,z$) when referring to a single variable in a polynomial ring, and by uppercase letters (such as $X_n, Z_n, Z$) when referring to variables in polynomial rings which can also be understood as random variables with some distribution.

To work with such multivariate monomials, we introduce the following definitions:

\begin{defin} \label{index_def}
A finitely supported function $\beta:\Zd\to\N\cup\{0\}$ will be called a multi-index. Let $\beta_n$ denote the value $\beta(n)$ for every $n\in\Zd$. Let $X^\beta$ denote the monomial $\prod_{n\in\Zd}X_n^{\beta_n}$.
\end{defin}

Fix a multi-index $\delta$, by
$$\delta_n=
\begin{cases}
1 &n=0\\
0 & n \neq 0
\end{cases}
$$
\begin{defin} \label{shifting_def}
For every multi-index $\beta$ and $i\in\Zd$, define $\beta^i$ $($$\beta$ shifted by $i$$)$, by $\beta^i_n=\beta_{n-i}$ for every $n\in\Zd$.
\end{defin}
Note that using these definitions, for $n,i\in\Zd$, $\delta^i_n$ is $1$ if $n=i$ and $0$ otherwise. Additionally, $\beta=\sum_{i\in\Zd} \beta_i \delta^i$ for every multi-index $\beta$ (this is a finite sum as $\beta$ is finitely supported).

Next, we fix $k\in\N$ and begin exploring the asymptotic behavior (as $L\rightarrow\infty$) of $\Tr{H_L^k}$. As we shall see, the coefficient of any monomial $X^\beta$ in $\Tr{H_L^k}$ is fixed for sufficiently large $L$, and has a concrete combinatorial description. Furthermore, these coefficients are invariant under translations of the monomials in $\Zd$. The precise statement is given in Proposition \ref{coef_prop} below, which requires some more definitions.

\begin{defin} \label{string_def}
Let $\mathcal S=\Set{V,U_1,U_2,\ldots,U_d,D_1,D_2,\ldots,D_d}$ be considered as formal symbols.\\
Then $\mathcal S^k$ denotes the set of all ordered $k$-tuples with elements from $\mathcal S$, or all strings of length $k$ from the alphabet $\mathcal S$.
\end{defin}
\begin{defin} \label{path_def}
For every $s\in\mathcal S^k$, we define a finite sequence of points, $y_0(s), y_1(s),\ldots,y_k(s)\in\Zd$ as follows:
\begin{itemize}
\item $y_0(s)=(0,0,\ldots,0)$,
\item $y_j(s)=
\begin{cases}
y_{j-1}(s)+e_v & s_j=U_v \\
y_{j-1}(s)-e_v & s_j=D_v \\
y_{j-1}(s)& s_j=V.
\end{cases}$
\end{itemize}
We say that $s$ is \emph{balanced}, if $y_k(s)=y_0(s)$.
\end{defin}
Note that $s\in\mathcal S^k$ is balanced iff for every $v=1,2,\ldots,d$, the symbols $U_v$ and $D_v$ appear in $s$ the same number of times.

\begin{defin} \label{corresponding_index_def}
For every $s\in\mathcal S^k$, define a multi-index $\varphi(s)$ by
$$\varphi(s)_n=\#\SetPred{1\leq j\leq k}{y_j(s)=y_{j-1}(s)=n},$$
for every $n\in\Zd$.
\end{defin}

\begin{defin} \label{path_counting_def}
For every multi-index $\beta$, let $\path^k(\beta)$ be the number of balanced strings $s\in\mathcal S^k$ satisfying $\varphi(s)^i=\beta$, for some $i\in\Zd$.
\end{defin}
Note that for every $s\in\mathcal S^k$ and multi-index $\beta$, there is at most one $i\in\Z^d$ for which $\varphi(s)^i=\beta$.

\begin{propos} \label{coef_prop}
For every non-zero multi-index $\beta$, and $k,L\in\N$, let $a^k_L(\beta)$ denote the coefficient of $X^\beta$ in the polynomial $\Tr{H_L^k}$. Then:
\begin{enumerate}
\item $0\leq a^k_L(\beta)\leq \path^k(\beta)$,
\item If $\beta_n>0$ for some $n\in\Lambda_{L-k}^d$, we have $a^k_L(\beta)=\path^k(\beta)$,
\item If $\beta_n>0$ for some $n\notin \boxL$, we have $a^k_L(\beta)=0$.
\end{enumerate}
\end{propos}

\begin{proof}
Use (\ref{operator_decomp}) to expand $H^k$. This gives us a bijection between operators in the expansion of $H^k$ and strings in $\mathcal S^k$. Furthermore, let $M_L$ be any matrix in the expansion of $H_L^k$ corresponding to a string $s\in\mathcal S^k$. It is straightforward to verify that if $s\in\mathcal S^k$ is balanced, and $i\in\boxL$, and $y^j(s)+i\in \boxL$ for every $j=1,2,\ldots,k$, we have $\Par{M_L}_{i,i}=X^{\varphi(s)^i}$. Otherwise, we have $\Par{M_L}_{i,i}=0$.\\
Therefore, fixing a multi-index $\beta$, the coefficient $a^k_L(\beta)$ equals the number of strings $s\in\mathcal S^k$, for which $\varphi(s)^i=\beta$ and the additional conditions $y_j(s)+i\in\boxL$ are fulfilled (we simply compute the trace as the sum over all diagonal entries from all matrices in the expansion). The number of such strings is at least $0$ and at most $\path^k(\beta)$ (which is the number of such strings without the additional conditions), proving (1).\\
Note that for any balanced $s\in\mathcal S^k$, we have $\abs{y_j(s)}\leq \frac k2$ for every $j=0,1,\ldots,k$. We deduce that whenever $\beta$ takes a non-vero value in $\Lambda_{L-k}^d$, if $\beta=\varphi(s)^i$ we must have $y_j(s)+i\in\Lambda_{L-k}^d$ for \emph{some} $j$, therefore $y_j(s)+i\in\boxL$ for \emph{every} $j=0,1,\ldots,k$. For such $\beta$, any $i\in\Zd$ and $s\in\mathcal S^k$ satisfying $\varphi(s)^i=\beta$ automatically fulfill the additional conditions, proving (2).\\
Similarly, if $\beta$ obtains a non-zero value outside of $\boxL$, satisfying $\varphi(s)^i=\beta$ guarantees that $y_j(s)+i\notin\boxL$ for some $j$, therefore $X^\beta$ doesn't appear anywhere on the diagonal of $M_L$, proving (3).
\end{proof}

Note that from definition \ref{path_counting_def}, it is clear that $\path^k\Par{\beta^i}=\path^k(\beta)$, for any multi-index $\beta$, any $i\in\Zd$, and any $k\in\N$.
Therefore, when considering the integers $\path^k(\beta)$ which appear as coefficients in the polynomials $\Tr{H_L^k}$, we may restrict our attention to a set of non-zero multi-indices which contains some shifting of every multi-index exactly once.
We denote this set by $B$:
\begin{defin} \label{index_set_def}
Two multi-indices $\beta$ and $\gamma$ are said to be equivalent if $\gamma=\beta^i$ for some $i\in\Zd$. From each equivalence class other than zero, choose a unique representative $\beta$, satisfying $\beta_0>0$ $($one way to make such choices, is to require the lexicographic minimum of the support of $\beta$ to be $0$$)$. Let $B$ be the set of all chosen representatives.
\end{defin}
In other words, $B$ is any set of multi-indices with the properties:
\begin{enumerate}
\item For any non-zero multi-index $\gamma$, we have $\gamma^i\in B$ for a unique $i\in\Zd$.
\item $\beta_0>0$ for every $\beta\in B$.
\end{enumerate}

\section{A central limit theorem for polynomial linear statistics}
In this section, we prove that for every polynomial $\poly(x) \in \R[x]$,
$$\frac{\Tr{\poly(H_L)}-\Ex{\Tr{\poly(H_L)}}}{\Par{2L+1}^{d/2}}$$
converges in distribution (as $L\rightarrow\infty$) to a normal distribution with variance $\sigma(\poly)^2\in[0,\infty)$ (see Proposition \ref{poly_CLT} below). We start by proving this CLT in the case where $\poly(x)=x^k$ is a monomial, which is easier to prove for an approximated version of the random variable $\Tr{H_L^k}$:
\begin{defin} \label{approximating_def}
For every $k,L\in\N$, let
\begin{equation} \label{approximating_def_eq}
T_L^k=\sum_{\beta\in B}\path^k(\beta) \sum_{i\in\boxL} X^{\beta^i},
\end{equation}
which we consider both as a random variable, and as a polynomial in the variables $\SetPred{X_n}{n\in\Zd}$.
\end{defin}
Note that the above sum is finite, since $\path^k(\beta)=0$ for all but finitely many $\beta\in B$.\\
We start be proving that $T_L^k$ can indeed approximate $\Tr{H_L^k}$, in the following sense:
\begin{propos} \label{approximating_prop}
For every $k\in\N$, the random variables
$$\frac{\Tr{H_L^k}-\Ex{\Tr{H_L^k}}}{\Par{2L+1}^{d/2}} - \frac{T_L^k-\Ex{T_L^k}}{\Par{2L+1}^{d/2}}$$
converge in probability $($as $L\to\infty$$)$ to $0$.
\end{propos}
\begin{proof}
It is sufficient to show that $\Var{T_L^k-\Tr{H_L^k}}=o\Par{L^d}$.
From Proposition \ref{coef_prop} and (\ref{approximating_def}), we have
$$T_L^k-\Tr{H_L^k}=\sum_{\beta\in B} \sum_{i\in\boxL} \Par{\path^k(\beta)-a^k_L\Par{\beta^i}}X^{\beta^i},$$
where $\path^k(\beta)-a^k_L\Par{\beta^i}=0$ whenever $i\in\Lambda_{L-k}^d$. Therefore, the number of non-zero terms in the above sum is at most
$$\abs{B_k}\Par{\abs{\Lambda_L^d}-\abs{\Lambda_{L-k}^d}}=O\Par{L^{d-1}},$$
where $B_k=\SetPred{\beta\in B}{\path^k(\beta)\neq 0}$ is finite. Next, consider the sum
\begin{equation} \label{approximation_var_eq}
	\begin{split}
&\Var{T_L^k-\Tr{H_L^k}}=\\
&\sum_{\beta,\gamma\in B_k} \sum_{i,j\in\boxL} \Par{\path^k(\beta)-a^k_L\Par{\beta^i}}\Par{\path^k(\gamma)-a^k_L\Par{\gamma^j}}\Cov{X^{\beta^i} , X^{\gamma^j}}.
	\end{split}
\end{equation}
Fixing $\beta,\gamma\in B_k$ and $i\in\Lambda_L^d\setminus\Lambda_{L-k}^d$, we see that whenever $j-i\notin\Lambda_k^d$, the supports of $\beta^i$ and $\gamma^j$ are disjoint, therefore $X^{\beta^i}$ and $X^{\gamma^j}$ are independent. This tells us that there are at most
$$\abs{B_k}^2\cdot \Par{\abs{\Lambda_L^d} - \abs{\Lambda_{L-k}^d}} \cdot \abs{\Lambda_k^d}=O\Par{L^{d-1}}$$
non-zero terms in (\ref{approximation_var_eq}). We know from (1) of Proposition \ref{coef_prop} that
\begin{equation} \label{path_mult_bound}
0\leq\Par{\path^k(\beta)-a^k_L\Par{\beta^i}}\Par{\path^k(\gamma)-a^k_L\Par{\gamma^j}}\leq \path^k(\beta)\path^k(\gamma).
\end{equation}
Since $\SetPred{X_n}{n\in\Zd}$ are identically distributed, we have $\Cov{X^{\beta^i} , X^{\gamma^j}} = \Cov{X^\beta , X^{\gamma^{j-i}}}$. From here we deduce that for fixed $\beta,\gamma\in B_k$, the term $\Cov{X^{\beta^i} , X^{\gamma^j}}$ only obtains a finite number of values: it is either $0$ or uniquely determined by $\beta,\gamma\in B_k$, the value of $j-i\in\Lambda_k^d$, and some of the (finite) moments of the underlying distribution $\textnormal{d}\rho$. Together with (\ref{path_mult_bound}), this gives us a uniform bound on all terms in (\ref{approximation_var_eq}), showing that indeed
$$\Var{T_L^k-\Tr{H_L^k}}=O\Par{L^{d-1}}.$$
\end{proof}

Our next step is a central limit theorem for the random variables $T_L^k$. Although our initial random variables $\SetPred{X_n}{n\in\Zd}$ were iid, for a fixed multi-index $\beta$ the random variables $\SetPred{X^{\beta^i}}{i\in\Zd}$ are generally not independent.
However, for any $i$ and $j$ sufficiently far apart ($j-i\notin\Lambda_k^d$ is sufficient), the variables $X^{\beta^i}$ and $X^{\beta^j}$ are independent. We use CLTs for weakly dependent random variables, by Hoeffding and Robbins \cite{HR}, and Neumann \cite{Neumann}, to prove:

\begin{thm} \label{main_CLT}
Let $B$ be any set of multi-indices such that $\beta_0>0$ for every $\beta\in B$. Let $\Set{a_\beta}_{\beta\in B}$ be a set of coefficients, such that $a_\beta=0$ for all but finitely many $\beta\in B$. Then
$$\frac 1{\Par{2L+1}^{d/2}} \sum_{\beta\in B} a_\beta \sum_{i\in\boxL} \Par{X^{\beta^i}-\Ex{X^{\beta^i}}}\overset{d}{\longrightarrow} N\Par{0,\sigma^2},$$
as $L\rightarrow\infty$, for some $\sigma^2\geq 0$.
\end{thm}
The proof is postponed to the appendix (Section 5).

\begin{cor} \label{approximate_CLT}
For every $k\in\N$,
$$\frac{T_L^k-\Ex{T_L^k}}{\Par{2L+1}^{d/2}}\overset{d}{\longrightarrow} N\Par{0,\sigma_k^2}$$
as $L\rightarrow\infty$, for some $\sigma_k^2\geq 0$.
\end{cor}

Now that we have a central limit theorem for our approximating random variables, we would like to compute the limit variances, and more generally, the limit covariances. We do this first for individual multi-indices:

\begin{lem} \label{basic_cov_lem}
For every two multi-indices $\beta$ and $\gamma$, we have
\begin{equation} \label{basic_cov_eq}
\underset{L\to\infty}\lim \frac 1{\Par{2L+1}^d}\Cov{\sum_{i\in\boxL} X^{\beta^i}, \sum_{i\in\boxL} X^{\gamma^i}}=\sum_{j\in\Zd}\Cov{X^\beta,X^{\gamma^j}}
\end{equation}
\end{lem}

Note that the sum on the right hand side of (\ref{basic_cov_eq}) is uniquely determined by $\beta,\gamma$, and the moments of the underlying distribution $\textnormal{d}\rho$, and it is in fact a finite sum: since $\beta$ and $\gamma$ are finitely supported, the supports of $\beta$ and $\gamma^j$ are disjoint (and therefore $X^\beta$ and $X^{\gamma^j}$ are independent) for all but finitely many $j\in\Zd$.

\begin{proof}
Since $\SetPred{X_n}{n\in\Zd}$ are identically distributed, the covariances are invariant to translations, and we may write
\begin{equation} \nonumber
	\begin{split}
\Cov{\sum_{i\in\boxL} X^{\beta^i}, \sum_{i\in\boxL} X^{\gamma^i}}&=\sum_{i,i'\in\boxL}\Cov{X^{\beta^i},X^{\gamma^{i'}}}\\
&=\sum_{i,i'\in\boxL}\Cov{X^\beta,X^{\gamma^{i'-i}}}\\
&=\sum_{j\in\Zd}z_j(L)\cdot \Cov{X^\beta,X^{\gamma^j}},
	\end{split}
\end{equation}
where $z_j(L)=\#\SetPred{i,i'\in\boxL}{j=i-i'}$. Clearly
$$\limL  \frac{z_j(L)}{\Par{2L+1}^d}=1$$
for any $j\in\Zd$, and since $\Cov{X^\beta,X^{\gamma^j}}\neq 0$ only for finitely many $j\in\Zd$, the claim follows.
\end{proof}

\begin{cor} \label{approximate_cov}
For every $k,\ell\in\N$,
\begin{equation}
	\begin{split}
&\limL \Cov{\frac{T^k_L}{\Par{2L+1}^{d/2}} , \frac{T^\ell_L}{\Par{2L+1}^{d/2}}}=\\
&\limL \frac 1{\Par{2L+1}^d}\Cov{T^k_L , T^\ell_L}=\sum_{\beta,\gamma \in B} \path^k(\beta)\path^\ell(\gamma) \sum_{j\in\Zd} \Cov{X^\beta , X^{\gamma^j}}
	\end{split}
\end{equation}
\end{cor}


This allows us to deduce results for the asymptotic behavior of the trace of monomials:
\begin{cor} \label{monom_CLT}
For any $k\in\N$,
$$\frac{\Tr{H_L^k}-\Ex{\Tr{H_L^k}}}{\Par{2L+1}^{d/2}} \overset{d}{\longrightarrow} N\Par{0,\sigma_k^2}$$
as $L\rightarrow\infty$, where
$$\sigma_k^2=\sum_{\beta,\gamma\in B} \path^k(\beta)\path^k(\gamma) \sum_{j\in\Zd} \Cov{X^\beta , X^{\gamma^j}}.$$
Furthermore, for every $k,\ell\in\N$,
\begin{equation} \nonumber
	\begin{split}
&\limL \Cov{\frac{\Tr{H_L^k}-\Ex{\Tr{H_L^k}}}{\Par{2L+1}^{d/2}} , \frac{\Tr{H_L^\ell}-\Ex{\Tr{H_L^\ell}}}{\Par{2L+1}^{d/2}}}=\\
&\sum_{\beta,\gamma\in B} \path^k(\beta)\path^\ell(\gamma)\sum_{j\in\Zd} \Cov{X^\beta , X^{\gamma^j}}.
	\end{split}
\end{equation}
\end{cor}
\begin{proof}
Follows directly from Proposition \ref{approximating_prop} and Corollaries \ref{approximate_CLT}, \ref{approximate_cov}.
\end{proof}

And now we can prove the CLT for any polynomial:
\begin{propos} \label{poly_CLT}
Let $\poly(x)=\sum_{k=0}^m a_k x^k\in\R[x]$ be a polynomial. Then
$$\frac{\Tr{\poly(H_L)}-\Ex{\Tr{\poly(H_L)}}}{\Par{2L+1}^{d/2}} \overset{d}{\longrightarrow} N\Par{0,\sigma(\poly)^2}$$
as $L\rightarrow\infty$, where
$$\sigma(\poly)^2=\sum_{k,\ell=1}^m a_k a_\ell \sum_{\beta,\gamma \in B} \path^k(\beta)\path^\ell(\gamma) \sum_{j\in\Zd} \Cov{X^\beta , X^{\gamma^j}}.$$
\end{propos}
\begin{proof}
Since $\Tr{\poly(H_L)}-\Ex{\Tr{\poly(H_L)}}$ doesn't depend on $a_0$, we may assume w.l.o.g. that $a_0=0$ and $\deg(\poly)=m>0$. Then
$$
\frac{\Tr{\poly(H_L)}-\Ex{\Tr{\poly(H_L)}}}{\Par{2L+1}^{d/2}}=\sum_{k=1}^m a_k \frac{\Tr{H_L^k}-\Ex{\Tr{H_L^k}}}{\Par{2L+1}^{d/2}},
$$
and from Corollary \ref{monom_CLT} we obtain the value of the variance $\sigma(\poly)^2$. Using Proposition \ref{approximating_prop} and (\ref{approximating_def_eq}), we now rewrite
\begin{equation} \label{limit_decomp}
	\begin{split}
&\limL \frac{\Tr{\poly(H_L)}-\Ex{\Tr{\poly(H_L)}}}{\Par{2L+1}^{d/2}}=
\limL \frac 1{\Par{2L+1}^{d/2}} \sum_{k=1}^m a_k \Par{T_L^k-\Ex{T_L^k}}=\\
&\limL \frac 1{\Par{2L+1}^{d/2}} \sum_{k=1}^m a_k \sum_{\beta\in B} \path^k(\beta) \sum_{i\in\boxL} \Par{X^{\beta^i}-\Ex{X^{\beta^i}}}=\\
&\limL \frac 1{\Par{2L+1}^{d/2}}
\sum_{\beta\in B}\Par{\sum_{k=1}^m a_k \path^k(\beta)} \sum_{i\in\boxL} \Par{X^{\beta^i}-\Ex{X^{\beta^i}}}.
	\end{split}
\end{equation}
Note that the first equality holds in the sense that both limit random variables have the same distribution.
Theorem \ref{main_CLT} now applies, proving that the limit has a normal distribution.
\end{proof}

\section{Degenerate and non-degenerate cases}
Now that we proved the convergence in Theorem \ref{main_thm}, it remains to determine under which conditions the limit distribution is non-degenerate, that is when $\sigma(\poly)^2>0$ for a non-constant polynomial $\poly\in\R[x]$. It turns out that $\sigma(\poly)^2$ is always positive if $\deg(\poly)\neq 2,3,5$, but for some polynomials of degree $2,3,5$ and some specific underlying distributions, the variance may vanish. We first demonstrate positive variance in degrees $\neq 2,3,5$:

\begin{propos} \label{deg_good_var}
Let $\poly(x)=\sum_{k=0}^m a_k x^k\in\R[x]$ be a non-constant polynomial of degree $m\neq 2,3,5$. Then $\sigma(\poly)^2>0$.
\end{propos}
\begin{proof}
Using (\ref{limit_decomp}), we write:
\begin{equation} \label{limit_var}
	\begin{split}
&\sigma(\poly)^2=\Var{\limL \frac{\Tr{\poly(H_L)}-\Ex{\Tr{\poly(H_L)}}}{\Par{2L+1}^{d/2}}}=\\
&\limL \Var{\frac 1{\Par{2L+1}^{d/2}} \sum_{\beta\in B}\Par{\sum_{k=1}^m a_k \path^k(\beta)} \sum_{i\in\boxL} \Par{X^{\beta^i}-\Ex{X^{\beta^i}}}}.
	\end{split}
\end{equation}
We follow the same general method used in \cite{BGW} - it is sufficient to find a multi-index $\abeta\in B$, with the following properties:
\begin{enumerate}
	\item $\path^m(\abeta)\neq 0$.
	\item $\path^k(\abeta)=0$ for every $k<m$.
	\item $\sum_{j\in\Zd}\Cov{X^\abeta,X^{\abeta^j}}>0$.
	\item $\Cov{X^\abeta, X^{\beta^j}}=0$ for every $j\in\Zd$ and every $\abeta\neq\beta\in B$ satisfying $\path^k(\beta)\neq 0$ for some $1\leq k\leq m$.
\end{enumerate}
If we find such $\abeta$, we deduce from property (4) that the random variables
$$Y_L^1\equiv a_m \path^m(\abeta)\sum_{i\in\boxL} \Par{X^{\abeta^i}-\Ex{X^{\abeta^i}}}$$
and
$$Y_L^2\equiv \sum_{\beta\in B\setminus\{\abeta\}} \Par{\sum_{k=1}^m a_k \path^k(\beta)} \sum_{i\in\boxL} \Par{X^{\beta^i}-\Ex{X^{\beta^i}}}$$
are uncorrelated (for any $L\in\N$), and (\ref{limit_var}) becomes
\begin{equation} \nonumber
	\begin{split}
\sigma(\poly)^2&=\limL \Var{\frac {Y_L^1+Y_L^2}{\Par{2L+1}^{d/2}}}\\
&=\limL \Var{\frac {Y_L^1}{\Par{2L+1}^{d/2}}}+\limL \Var{\frac {Y_L^2}{\Par{2L+1}^{d/2}}}\\
&\geq\limL \Var{\frac {Y_L^1}{\Par{2L+1}^{d/2}}}=a_m \path^m(\abeta) \sum_{j\in\Zd}\Cov{X^\abeta , X^{\abeta^j}}>0,
	\end{split}
\end{equation}
where the final equality is due to Lemma \ref{basic_cov_lem}. We make the following choices for $\abeta$:
\begin{enumerate}
	\item If $m=1$, choose $\abeta=\delta$.
	\item If $m\geq 4$ is even, choose $\abeta=\delta+\delta^{\Par{\frac m2-1}e_1}$.
	\item If $m\geq 7$ is odd, choose $\abeta=\delta+\delta^{e_1}+\delta^{\Par{\frac{m-3}{2}}e_1}$.
\end{enumerate}
The proof that these $\abeta$ satisfy properties (1) and (2) is straightforward path counting. For (3) and (4), recall that
\begin{equation} \nonumber
\begin{split}
\Cov{X^\abeta, X^{\beta^j}}&=\Ex{X^{\abeta}X^{\beta^j}}-\Ex{X^\abeta}\Ex{X^{\beta^j}}\\
&=\Ex{\prod_{n\in\Zd}X_n^{\abeta_n}X_n^{\beta_n^j}}-\Ex{\prod_{n\in\Zd}X_n^{\abeta_n}}\Ex{\prod_{n\in\Zd}X_n^{\beta_n^j}}\\
&=\prod_{n\in\Zd}\Ex{X_n^{\abeta_n+\beta^j_n}}-\prod_{n\in\Zd}\Ex{X_n^{\abeta_n}}\Ex{X_n^{\beta^j_n}}.
\end{split}
\end{equation}
If there exists any $n\in\Zd$ such that $\gamma_n=1$ and $\beta^j_n=0$, the term $\Ex{X_n}=0$ appears in both products, thus $\Cov{X^\abeta , X^{\beta^j}}=0$. Thus any $\beta^j$ for which $\Cov{X^\abeta, X^{\beta^j}}\neq 0$ must have $\beta^j_n\geq\abeta_n$ for every $n\in\Zd$.
If $\beta^j=\abeta$, since $\beta,\abeta\in B$ we must also have $j=0$ and $\beta=\abeta$. Otherwise, we have $\beta^j_n>\abeta_n$ for some $n\in\Zd$, and it is straightforward to verify that every string $s$ with $\varphi(s)^i=\beta^j$ must have length $>m$, therefore $p^k(\beta)=0$ for every $1\leq k\leq m$.

Note that there is some freedom in the choice of the representative set $B$, but one may choose $B$ such that $\abeta\in B$ in all of the above cases, or alternatively replace the above choice of $\abeta$ with some $\abeta^i\in B$.
\end{proof}

For polynomials $\poly$ of degree $2,3$ or $5$, we must carefully analyze all cases. Since there are specific underlying distributions and polynomials $\poly$ for which $\sigma(\poly)^2=0$, and we want an explicit description of all such cases, we need to explicitly compute all non-zero values of $\path^k(\beta)$, for $1\leq k\leq 5$.

\begin{lem} \label{coef_values}
If $k\in\Set{1,2,\ldots,5}$ and $\gamma$ is a multi-index with $\path^k(\gamma)>0$, then
\begin{enumerate}
\item $\gamma$ is either equivalent to a unique $\beta$ which equals $m\cdot\delta$ $($for some $m\in\Set{1,2,\ldots,5}$$)$, or to one of $\delta+\delta^e$, $2\delta+\delta^e$, or $2\delta+\delta^{-e}$ $($for some $e\in \Set{e_1,e_2,\ldots,e_d}$$)$.
\item The value of $\path^k(\gamma)=\path^k(\beta)$ is given in the table below $($empty entries correspond to $\path^k(\beta)=0$$)$:
$$\begin{array}{|c|c|c|c|c|c|c|c|}
\hline
\raisebox{-.2em}{$k$} \diagdown \raisebox{.2em}{$\beta$} &\delta &2\delta &3\delta &4\delta &5\delta &\delta+\delta^e &2\delta + \delta^{\pm e} \\
\hline
1 &1 & & & & & & \\
\hline
2 & &1 & & & & & \\
\hline
3 &6d & &1 & & & & \\
\hline
4 & &8d & &1 & &4 & \\
\hline
5 &60d^2-30d & &10d & &1 & &5 \\
\hline
\end{array}$$
\end{enumerate}
\end{lem}
To prove the lemma, we found no alternative to enumerating the relevant strings in $\mathcal S^k$ (for $k=1,2,\ldots,5$). We omit this technical proof.

Our method of verifying which polynomials $\poly(x)=\sum_{k=0}^5 a_k x^k$ satisfy $\sigma(\poly)^2>0$, is to describe random variables $W_1,W_2,\ldots,W_5$ such that \\
$\Var{\sum_{k=1}^5 a_k W_k}=\sigma(\poly)^2$ (for any choice of coefficients $a_0,a_1,\ldots,a_5$). We then explore the random variable $\sum_{k=1}^5 a_k W_k$ and determine under which conditions it is almost surely constant.\\
If $\deg(\poly)\leq 3$, we may replace $\Set{W_i}$ with a simpler set of random variables, $T_1,T_2,T_3$. We verify this case before approaching polynomials of degree 5:

\begin{propos} \label{deg_23_var}
Let $\poly(x)=\sum_{k=0}^m a_k x^k\in\R[x]$ be a polynomial of degree $1\leq m\leq 3$. Then:
\begin{enumerate}
\item If the underlying distribution $(\textnormal{d}\rho)$ is supported by more than three values, then $\sigma(\poly)^2>0$.
\item If the underlying distribution is supported by exactly three values, denoted $a,b,c\in\R$, then $\sigma(\poly)^2=0$ iff $\poly=a_3 \tildq_3 + a_0$, where
$$\tildq_3(x)=x^3-(a+b+c)x^2+(ab+ac+bc-6d)x.$$
\item If the underlying distribution is supported by exactly two values, denoted $a,b\in\R$, then $\sigma(\poly)^2=0$ iff $\poly=a_3 \polyb_3 + a_2 \polyb_2 + a_0$, where
$$\polyb_3(x)=x^3 - (a^2+ab+b^2+6d)x,\qquad \polyb_2(x)=x^2 - (a+b)x.$$
\end{enumerate}
\end{propos}

\begin{proof}
Let $Z$ denote both a random variable distributed by $\textnormal{d}\rho$, and the variable in polynomial ring $\R[Z]$. Define
$$T_1=Z,\qquad T_2=Z^2,\qquad T_3=Z^3+6dZ.$$
Using Lemma \ref{coef_values}, we see that for every $k=1,2,3$, we have\\
 $T_k=\sum_{n=1}^3 \path^k(n\delta) Z^n$, thus for every $k,\ell=1,2,3$:
\begin{equation} \nonumber
\begin{split}
\Cov{T_k , T_\ell}=& \sum_{n,m=1,2,3} \path^k(n\delta) \path^\ell(m\delta) \Cov{Z^n, Z^m}\\
& \sum_{\beta,\gamma=\delta,2\delta,3\delta} \path^k(\beta) \path^\ell(\gamma) \Cov{X^{\beta}, X^{\gamma}}\\
& \sum_{\beta,\gamma\in B} \path^k(\beta) \path^\ell(\gamma) \sum_{j\in\Zd} \Cov{X^\beta, X^{\gamma^j}}
\end{split}
\end{equation}
(we may assume w.l.o.g. that $\delta,2\delta,3\delta\in B$).
We now deduce from Proposition \ref{poly_CLT} that
$$\sigma(\poly)^2=\Var{a_3 T_3 + a_2 T_2 + a_1 T_1},$$
which is zero iff $\Poly=a_3 T_3 + a_2 T_2 + a_1 T_1$ is almost surely constant, as a random variable. 
As a polynomial, $\Poly\in\R[Z]$ has at most $3$ distinct roots, so if $Z$ is supported by more than $3$ points, $\Poly$ is non-constant as a random varaible, thus
$\Var{\Poly}>0$, proving (1).\\
Observe that any assignment of a value to the random variable $Z$ corresponds to a ring homomorphism $\R[Z]\rightarrow\R$. Furthermore, if we only assign values from $\Set{a,b,c}$, all three assignment homomorphisms factor through the quotient ring $\R[Z]/\Par{(Z-a)(Z-b)(Z-c)}$. Write $P\equiv Q$ for two polynomials $P,Q$, if they have the same projection in the quotient. Note that $P\equiv Q$ iff as random variables, $P=Q$ almost surely. Clearly $\Var{\Poly}=0$ as a random variable iff $\Poly\equiv \text{const}$ in $\R[Z]$. Now write
\begin{equation} \nonumber
	\begin{split}
(Z-a)(Z-b)(Z-c)&=Z^3-(a+b+c)Z^2+(ab+ac+bc)Z-abc\\
&=T_3-(a+b+c)T_2+(ab+ac+bc-6d)T_1-abc,
	\end{split}
\end{equation}
and deduce (2): the polynomial $\tildq_3(x)$ has $\sigma(\tildq_3)=0$ from the above, therefore $\sigma(a_3\tildq_3 + a_0)=0$. If $\poly\neq a_3\tildq_3+a_0$, we see that $\Poly=a_3T_3+a_2T_2+a_1T_2$ is equivalent to a polynomial of degree $1$ or $2$ in $\R[Z]$, and therefore isn't fixed under assignments from $\Set{a,b,c}$.\\
Finally, if $\textnormal{d}\rho$ is supported on $\Set{a,b}$, the same arguments hold with a different quotient ring, $\R[Z]/\Par{(Z-a)(Z-b)}$. Now note that
$$0\equiv (Z-a)(Z-b)=Z^2-(a+b)Z+ab,$$
therefore
$$T_2-(a+b)T_1=Z^2- (a+b)Z\equiv \text{const},$$
proving $\sigma(\polyb_2)^2=0$. We also have
$$Z^3\equiv (a+b)Z^2-abZ\equiv (a^2+ab+b^2)Z-ab(a+b),$$
therefore
$$T_3-(a^2+ab+b^2+6d)T_1=Z_3-\Par{a^2+ab+b^2}Z \equiv \text{const},$$
proving $\sigma(\polyb_3)^2=0$. If $\poly\neq a_3 \polyb_3 + a_2 \polyb_2 + a_0$, then the above computations show that $\Poly$ is equivalent to a polynomial of degree $1$, which isn't equivalent to any constant, therefore $\Var{\Poly}>0$.
\end{proof}

\begin{propos} \label{deg_5_var}
Let $\poly(x)=\sum_{k=0}^5 a_k x^k\in\R[x]$ be a polynomial of degree $5$. Then:
\begin{enumerate}
\item If the underlying distribution $(\textnormal{d}\rho)$ is supported by more than two values, then $\sigma(\poly)^2>0$.
\item If the underlying distribution is supported by exactly two values, denoted $a,b\in\R$, then $\sigma(\poly)^2=\sigma(\poly-a_5 \polyb_5)^2$, where
\begin{equation} \nonumber
	\begin{split}
&
\text{\scalebox{0.92}
{$2\polyb_5(x)=2x^5-5(a+b)x^4+$}}\\
&
\text{\scalebox{0.92}
{$\left[3(a^4+b^4)+8(a^3 b+a^2 b^2+a b^3)+20d(a^2+b^2)+100dab-120d^2+60d\right]x$}}.
	\end{split}
\end{equation}
In particular, $\sigma(\polyb_5)^2=0$.
\end{enumerate}
\end{propos}
\begin{proof}
For every $n\in\Lambda_1^d$, let $Z_n=X_n$. We regard the variables $\Set{Z_n}$ both as $3^d$ independent random variables distributed by $\textnormal{d}\rho$, and as the variables in the polynomial ring
$R=\R\SPar{Z_n\ \middle|\ n\in\Lambda_1^d}$. Define:
$$W_1=3^{-d/2}\sum_{n\in\Lambda_1^d} Z_n,\qquad W_2=3^{-d/2}\sum_{n\in\Lambda_1^d} Z_n^2.$$
and
$$W_3=3^{-d/2}\sum_{n\in\Lambda_1^d} \Par{Z_n^3+6dZ_n}.$$
Let $E$ consist of all unordered pairs $\Set{n,m}$, such that $n,m\in\Lambda_1^d$ differ in exactly one coordinate, that is
$$E=\SetPred{\Set{n,m}}{n,m\in\Lambda_1^d\ ,\ \#\SetPred{1\leq v\leq d}{n_v\neq m_v}=1}.$$
Now define
$$W_4=3^{-d/2}\sum_{n\in\Lambda_1^d} \Par{Z_n^4+8dZ_n^2}+3^{-d/2}\sum_{\Set{n,m}\in E} 4Z_nZ_m$$
and
\begin{equation} \nonumber
\begin{split}
W_5=&\ 3^{-d/2}\sum_{n\in\Lambda_1^d} \Par{Z_n^5+10dZ_n^3+\Par{60d^2-30d}Z_n}+\\
&\ 3^{-d/2}\sum_{\Set{n,m}\in E} 5\Par{Z_n^2 Z_m+Z_nZ_m^2}.
\end{split}
\end{equation}
Following Lemma \ref{coef_values} and a straightforward computation that we omit, we verify that
$$\Cov{W_k,W_\ell}=\sum_{\beta\in B}\path^k(\beta)\path^\ell(\gamma)\sum_{j\in\Zd}\Cov{X^\beta,X^{\gamma^j}}$$
for every $k,\ell\in\Set{1,2,\ldots,5}$ then deduce from Proposition \ref{poly_CLT} that
$$\sigma(\poly)^2=\Var{\sum_{k=1}^5 a_k W_k}.$$
Denote
\begin{equation} \label{W_def}
\Poly=\sum_{k=1}^5 a_k W_k.
\end{equation}
As in the proof of Proposition \ref{deg_23_var}, we note that $\Var{\Poly}=0$ iff $\Poly$ is almost surely constant as a random variable.
This shows that $\Var{\Poly}>0$ if $\textnormal{d}\rho$ is not finitely supported: generally if $\poly\in\R\SPar{x_1,\ldots,x_m}$ is a non-constant multivariate polynomial, and $S$ is a set such that $\poly(s_1,\ldots,s_m)=0$ for every $s_1 ,\ldots,s_m\in S$, then straightforward induction on $m$ shows that $\abs{S}\leq \deg(\poly)$.\\
So assume henceforth that the variables $Z_n$ are supported by a finite set $\supp(\textnormal{d}\rho)\subset\R$. Denote $q(x)=\prod_{a\in \supp(\textnormal{d}\rho)}\Par{x-a}$, and let $Q_n=q(Z_n)\in R$, and let $I\subset R$ be the ideal generated by the polynomials $\Set{Q_n}_{n\in\Lambda_1^d}$. Every possible assignment of values to $\Set{Z_n}$ corresponds to a ring homomorphism $R\to\R$. If we only assign values from $\supp(\textnormal{d}\rho)$ the homomorphism factors through the quotient ring $R/I$. Write $P\equiv Q$ for two polynomials $P,Q$, if they have the same projection in the quotient. Note that $P\equiv Q$ in $R$ iff as random variables, $P=Q$ almost surely. Clearly $\Var{\Poly}=0$ as a random variable iff $\Poly\equiv\text{const}$ in $R$.

Next, we denote $\omega_1=\omega_2=\omega_3=0$,
$$\omega_4=3^{-d/2}\sum_{\Set{n,m}\in E} 4Z_nZ_m,$$
and
\begin{equation} \label{omega5_def}
\omega_5=3^{-d/2}\sum_{\Set{n,m}\in E} 5\Par{Z_n^2 Z_m+Z_n Z_m^2}
\end{equation}
(so each $\omega_k$ is the part of $W_k$ which is a sum of products involving more than one variable). Now rewrite (\ref{W_def}) as
$$\Poly=a_5\omega_5+a_4\omega_4+\sum_{k=1}^5 a_k\Par{W_k-\omega_k},$$
and note that if $\abs{\supp(\textnormal{d}\rho)}\leq k$ then $W_k-\omega_k$ is equivalent to a polynomial of degree lower than $\abs{\supp(\textnormal{d}\rho)}$: every term of the form $Z_n^k$ is equivalent to $Z_n^k-Z_n^{k-\abs{\supp(\textnormal{d}\rho)}}q\Par{Z_n}$, with degree strictly lower than $k$. Thus
$$\sum_{n\in\Lambda_1}Z_n^k\equiv \sum_{n\in\Lambda_1^d} Z_n^k-Z_n^{k-\abs{\supp(\textnormal{d}\rho)}}q\Par{Z_n},$$
and summing over $n\in\Lambda_1^d$ allows us to reduce $W_k-\omega_k$ to an equivalent combination of $W_1-\omega_1,\ldots,W_{k-1}-\omega_{k-1}$, to eventually obtain
\begin{equation} \label{red_deg}
\Poly\equiv\widetilde \Poly=a_5\omega_5+a_4\omega_4+\sum_{k=1}^{\abs{\supp(\textnormal{d}\rho)}-1} \widetilde a_k\Par{W_k-\omega_k}
\end{equation}
for some $\widetilde a_1,\ldots,\widetilde a_{\abs{\supp(\textnormal{d}\rho)}-1}\in\R$.
We are now ready to prove that $\sVar{\widetilde \Poly}>0$, whenever $\abs{\supp(\textnormal{d}\rho)}\geq 3$. Otherwise, $\sVar{\widetilde \Poly}=0$ implies that $\widetilde \Poly-c\in I$ for some constant $c$, so we can find polynomials $H_n\in R$, such that
\begin{equation} \label{var0_eq1}
\widetilde \Poly - c=\sum_{n\in\Lambda_1^d} H_n\cdot Q_n
\end{equation}
in $R$. Fix some $a\in A$, and let $\psi_a:R\to\R[x]$ be the ring homomorphism, defined by
$$\psi_a(Z_n)=\begin{cases} x & n=0 \\ a & n\neq 0
\end{cases}\ .$$
We have $\psi_a(Q_n)=q(a)=0$ for every $n\neq0$, so when we apply $\psi_a$ to (\ref{var0_eq1}), we obtain the equality
\begin{equation} \label{var0_eq2}
\psi_a(\widetilde \Poly)-c=h(x)q(x)
\end{equation}
in $\R[x]$, where $h(x)=\psi_a(H_0)$.
Note that $W_k-\omega_k$ has degree $k$ in $R$, therefore $\psi_a\Par{W_k-\omega_k}$ has degree at most $k$. Clearly $\psi_a\Par{a_5\omega_5+a_4\omega_4}$ has degree $2$, so from (\ref{red_deg}) the polynomial in the left hand side of (\ref{var0_eq2}) has degree strictly less than $\abs{\supp(\textnormal{d}\rho)}$. But $q(x)$ has degree $\abs{\supp(\textnormal{d}\rho)}$, so we must have $h(x)=0$ (otherwise the right hand side of (\ref{var0_eq2}) would have degree $\abs{\supp(\textnormal{d}\rho)}$ or higher). We deduce that $\psi_a(\widetilde \Poly)-c=0$ as a polynomial in $\R[x]$.

Since for every $n\in\Lambda_1^d$ there are $\#\SetPred{m}{\Set{n,m}\in E}=2d$ values of $m$ for which $5\cdot Z_n^2 Z_m$ appears in the sum (\ref{omega5_def}), the coefficient of $x^2$ in $\psi_a\Par{\omega_5}$ is $2d\cdot 3^{-d/2}\cdot5a$. We deduce that the coefficient of $x^2$ in $\psi_a(\widetilde \Poly)-c=0$ is
$$a_5\cdot 10d\cdot 3^{-d/2}\cdot a + c'=0,$$
where $c'$ doesn't depend on our choice of $a\in \supp(\textnormal{d}\rho)$. Since $a_5\neq 0$, there is at most one $a\in\R$ satisfying the above equation. However, for any $b\in \supp(\textnormal{d}\rho)$, applying $\psi_b$ to (\ref{var0_eq1}) allows us to obtain $a_5\cdot 10d\cdot 3^{-d/2}\cdot b + c'=0$, which is a contradiction. This concludes the proof of (1).\\
If $\supp(\textnormal{d}\rho)=\Set{a,b}$ then $q\Par{Z_n}=\Par{Z_n-a}\Par{Z_n-b}\in I$, therefore
\begin{equation} \label{deg_2_red}
Z_n^2\equiv\Par{a+b}Z_n-ab
\end{equation}
for every $n\in\Lambda_1^d$, thus (\ref{omega5_def}) becomes $\omega_5\equiv\frac52\Par{a+b}\omega_4-20dab W_1$, which allows us to deduce
\begin{equation} \label{mixed_red}
a_5\omega_5+a_4\omega_4\equiv-20a_5dab W_1
\end{equation}
whenever $a_4=-\frac52 \Par{a+b}a_5$.\\
Finally, from (\ref{deg_2_red}) we verify:
\begin{equation} \label{pure_red}
\begin{split}
Z_n^2\equiv& \Par{a+b}Z_n-ab\\
Z_n^3\equiv& \Par{a^2+ab+b^2}Z_n-ab\Par{a+b}\\
Z_n^4\equiv& \Par{a^3+a^2 b+a b^2+b^3}Z_n-ab\Par{a^2+ab+b^2}\\
Z_n^5\equiv& \Par{a^4+a^3 b+a^2 b^2+a b^3+b^4}Z_n-\text{const}.
\end{split}
\end{equation}
Summing over $n\in\Lambda_1^d$ allows us to reduce $3^{-d/2}\sum_n Z_n^k$ (for $k=2,3,4,5$) to equivalent expressions involving $W_1$ and constants, and along with (\ref{mixed_red}) and the definitions of $W_1,W_4,W_5$ we deduce
\begin{equation} \nonumber
\begin{split}
& 2W_5-5\Par{a+b}W_4+\text{const}\equiv\\
& \text{\scalebox{0.94}{$\Par{-3a^4-8a^3 b-8a^2 b^2-8a b^3-3b^4-20da^2-100dab-20db^2+120d^2-60d}W_1$}}.
\end{split}
\end{equation}
From here it follows that $\sigma\Par{\polyb_5}^2=0$ and that $\sigma\Par{\poly}^2=\sigma\Par{\poly-c\polyb_5}^2$ for any polynomial $\Poly$ and constant $c$, concluding our proof.
\end{proof}

\begin{proof}[Proof of Theorem \ref{main_thm}]
Given a polynomial $\poly(x)=\sum_{k=0}^m a_k x^k \in\R[x]$, we have
$$\frac{\Tr{\poly\Par{H_L}}-\Ex{\Tr{\poly\Par{H_L}}}}{(2L+1)^{d/2}} \overset{d}{\longrightarrow} N(0,\sigma(\poly)^2)$$
for $\sigma(\poly)^2\in[0,\infty)$ as $L\rightarrow\infty$, from Proposition \ref{poly_CLT}. From Propositions \ref{deg_good_var} and \ref{deg_23_var} we determine the cases in which $\sigma(\poly)^2>0$ whenever $\deg(\poly)\neq 5$. Finally, if $\deg(\poly)=5$, we know from proposition \ref{deg_5_var} that $\sigma(\poly)^2=\sigma(\poly-a_5 \polyb_5)^2$. If $\poly-a_5\polyb_5$ is non-constant and $\deg(\poly-a_5\polyb_5)$ is $1$ or $4$, we determine that $\sigma(\poly)^2=\sigma(\poly-a_5 \polyb_5)^2>0$ from proposition \ref{deg_good_var}, otherwise we use proposition \ref{deg_23_var} to determine the positivity.
\end{proof}

\section{Appendix - Proof of Theorem \ref{main_CLT}}
In the setting of Theorem \ref{main_CLT}, we consider a $d$-dimensional array of weakly dependent random variables. 
Explicitly, we prove a central limit theorem which is valid in the setting of $m$-dependent random variables, which we now define:

\begin{defin} \label{m_dep_d1}
Let $\Set{Y_i}_{i\in\Zd}$ be a sequence of random variables. We say that the sequence is $m$-dependent, if for any two finite sets of indices, $I,J\subset\Zd$ which satisfy $\abs{i-j}>m$ for every $i\in I$ and $j\in J$, the corresponding sets of random variables,
$$\Set{Y_i}_{i\in I}, \qquad \Set{Y_j}_{j\in J} $$
are independent.
\end{defin}

Note that this definition extends a notion of $m$-dependence from \cite{HR} defined for sequences of variables indexed by $\N$ (the definition of $m$-dependence in \cite{HR} is equivalent to $m$-dependence as defined above, when we take $d=1$ and $Y_i=0$ for every $i\notin\N$). In \cite{HR}, Hoeffding and Robbins proved the following central limit theorem:

\begin{thm}[Hoeffding-Robbins] \label{thm_HR}
Let $\Set{X_i}_{i\in\N}$ be an $m$-dependent\\
sequence of random variables satisfying
$\Ex{X_i}=0$ and $\Ex{\abs{X_i}^3}\leq R^3 <\infty$ for every $i\in\N$, and
$$\lim_{p\to\infty} p^{-1}\sum_{h=1}^p A_{i+h}=A$$
uniformly for all $i\in\mathbb N$, where
$$A_i=\Ex{X_{i+m}^2}+2\sum_{j=1}^m \Ex{X_{i+m-j}X_{i+m}}.$$
Then
$$\frac{X_1+\ldots+X_n}{n^\frac12} \overset{d}{\longrightarrow} N\Par{0,A}.$$
\end{thm}

Theorem \ref{thm_HR} allows us to deduce a central limit theorem for $d=1$, and the following theorem by Neumann \cite{Neumann} will allow us to prove an induction argument on $d$:

\begin{thm}[Neumann] \label{thm_Neumann}
Suppose that $\Set{X_{n,k}\ \middle|\ n\in\N\ ,\ k=1,2,\ldots,n}$ is a triangular scheme of random variables with $\Ex{X_{n,k}}=0$ and
$$\sum_{k=1}^n \Ex{X^2_{n,k}}\le C$$
for all $n,k$ and some $C<\infty$. We assume that
$$\sigma_n^2 = \Var{X_{n,1}+...+X_{n,n}} \underset{n\rightarrow \infty}{\longrightarrow}\sigma^2\in [0,\infty), $$
and that
$$\sum_{k=1}^n \Ex{X_{n,k}^2 1(\abs{X_{n,k}}>\epsilon)}\underset{n \rightarrow \infty}{\longrightarrow}0 $$
holds for all $\epsilon>0$. Furthermore, we assume that there exists a summable sequence $(\theta_r)_{r\in \N}$ such that for all $u \in \N$ and all indices
$$1\le s_1<s_2<...<s_u<s_u+r=t_1\le t_2 \le n,$$
the following upper bounds for covariances hold true: for all measurable functions $g: \R^u \longrightarrow \R$ with $||g||_{\infty}=\sup_{x\in \R^u} |g(x)|\le 1$, we have
\begin{equation} \label{cov_req_1}
\abs{\Cov{g\Par{X_{n,s_1},...,X_{n,s_u}}X_{n,s_u}\ ,\ X_{n,t_1}}}\le \left(\Ex{X_{n,s_u}^2}+\Ex{X_{n,t_1}^2}+\frac{1}{n}\right)\theta_r
\end{equation}
and
\begin{equation} \label{cov_req_2}
\abs{\Cov{g\Par{X_{n,s_1},...,X_{n,s_u}}\ ,\ X_{n,t_1} X_{n,t_2}}}\le \left(\Ex{X_{n,t_1}^2}+\Ex{X_{n,t_2}^2}+\frac{1}{n}\right)\theta_r.
\end{equation}
Then
$$X_{n,1}+...+X_{n,n}\overset{d}{\longrightarrow} N\Par{0,\sigma^2}$$
as $n\rightarrow\infty$.
\end{thm}

Our central limit theorem for $m$-dependent random variables follows:

\begin{propos} \label{prop_aux_CLT}
Let $\Set{Y_i}_{i\in\Zd}$ be an identically distributed $d$-dimensional $m$-dependent array of random variables such that $\Ex{Y_i}=0$, and $\Ex{\abs{Y_i}^3}<\infty$.

    Then
$$\frac1{\Par{2L+1}^{d/2}}\sum_{i\in\boxL} Y_i \overset{d}{\longrightarrow} N\Par{0,\sigma^2},$$
where
$$\sigma^2=\lim_{L\to\infty}\frac1{\Par{2L+1}^d}\Var{\sum_{i\in\boxL}Y_i}.$$
\end{propos}

\begin{proof}
By induction on $d$. For $d=1$, this is a straightforward application of Theorem \ref{thm_HR} to the random variables $\Set{X_i}_{i\in\N}$, defined by $X_i=Y_{i+m}+Y_{-i-m}$ (noting that for $i>m$, $\Set{X_i}_{i\in\N}$ are identically distributed and $m$-dependent, and the exclusion of a finite set of random variables $\Set{Y_i\ :\ \abs{i}\leq m}$ from the sum has no effect on the limit distribution).\\

We now assume by induction that the proposition holds for some $d\in\N$, and prove it in dimension $d+1$. For every $L\in\N$ we denote $n=2L+1$, rewrite
$$\frac1{\Par{2L+1}^{\Par{d+1}/2}}\sum_{i\in\boxLp} Y_i=\sum_{j=-L}^L Z_{n,j},$$
where
$$Z_{n,j}=\frac1{n^{1/2}} \cdot \frac1{n^{d/2}}\sum_{i\in I_{n,j}} Y_i$$
and
\begin{equation} \nonumber
\begin{split}
I_{n,j}&=\boxL\times\Set{j}\\
&=\Set{\Par{i_1 ,\ldots,i_{d+1}}\in\boxLp\ \middle|\ i_{d+1}=j}
\end{split}
\end{equation}
are defined for every $j\in\Lambda_L$. Our proof will be completed by applying Theorem \ref{thm_Neumann} to the random variables
$$X_{n,k}=\begin{cases} Z_{n,k-L-1} & n=2L+1\\
Z_{n+1,k-L-1} & n=2L,\end{cases}$$
which are defined for every $n\in\N$ and $k=1,2,\ldots,n$. We will apply the requirements of the theorem to the corresponding variables $Z_{n,j}$ (we henceforth ignore even values of $n$).\\
Fixing any $j\in\Z$, we may identify $I_{n,j}$ with $\boxL$, and note that the $d$-dimensional array $\Set{Y_i\ \middle|\ i\in\Z^{d+1}\ ,\ i_{d+1}=j}$ is identically distributed and $m$-dependent (the distribution of the array is independent of $j\in\Z$ as well). The induction hypothesis now applies, and we deduce
\begin{equation} \label{ind_hyp}
\sqrt n Z_{n,j}=\frac1{n^{d/2}}\sum_{i\in I_{n,j}}Y_i \overset{d}{\longrightarrow} N\Par{0,\sigma_d^2}
\end{equation}
as $n\rightarrow\infty$, uniformly in $j$, for some $\sigma_d^2\geq0$. The variables $Z_{n,j}$ are ``well behaved", in the sense that for any sufficiently large $n$,
$$\Ex{Z_{n,j}^2}=\Var{Z_{n,j}}\leq\frac1n(\sigma_d^2+1)$$
(thus there exists $C>0$ such that $\Ex{Z_{n,j}^2}\le \frac{C}{n}$ for all $n\in\N$ and $j\in\Lambda_L$). We deduce that
$$\Ex{Z_{n,j}}=0,\qquad  \sum_{j=-L}^L \Ex{Z_{n,j}^2}\le C.$$
Additionally, since the finite sequence $\Set{Z_{n,j}}_{j\in\Lambda_L}$ is both identically distributed and $m$-dependent (for every $n=2L+1\in\N$), one can verify that
$$\Var{\sum_{j=-L}^L Z_{n,j}}\underset{n \rightarrow \infty}{\longrightarrow} \sigma^2<\infty.$$
Next, we prove that
$$\sum_{j=-L}^L \Ex{Z_{n,j}^2 1(\abs{Z_{n,j}}>\epsilon)}\underset{n \rightarrow \infty}{\longrightarrow}0$$
for every $\epsilon>0$. Note that
\begin{equation} \label{var_sum}
\begin{split}
\sum_{j=-L}^L \Ex{Z_{n,j}^2 1\Par{\abs{Z_{n,j}}>\epsilon}}&=n \Ex{Z_{n,j}^2 1\Par{\abs{Z_{n,j}}>\epsilon}}\\
&=\Ex{n \Par{Z_{n,j}}^2 1\Par{\abs{\sqrt n Z_{n,j}}>\epsilon\sqrt n}}.
\end{split}
\end{equation}
From the induction hypothesis, we know that $\sqrt{n}Z_{n,j}\overset{d}{\longrightarrow} N(0,\sigma_d^2)$. We deduce that for every $M>0$ we have
\begin{equation} \label{small_conv}
\sqrt{n}Z_{n,j}1(\abs{\sqrt n Z_{n,j}}>M)\overset{d}{\longrightarrow} \Phi_M,
\end{equation}
where $\Phi_M$ is a random variable satisfying $\Ex{\Phi_M}=0$, and
$$\Var{\Phi_M}= \begin{cases}
2\int_M^\infty \frac{t^2}{\sigma_d \sqrt{2\pi}}\exp{\Par{-\frac{t^2}{2\sigma_d^2}}}dt & \sigma_d^2>0\\
0 & \sigma_d^2=0.
\end{cases}$$
Choose some $M>0$ so that $\Var{\Phi_M}$ is arbitrarily close to $0$. For every $\epsilon>0$, any sufficiently large $n\in\N$ satisfies $\epsilon\sqrt n>M$, so
$$1\Par{\abs{\sqrt n Z_{n,j}}>\epsilon\sqrt n}\leq 1\Par{\abs{\sqrt n Z_{n,j}}>M},$$
and (\ref{var_sum}) now becomes
\begin{equation} \nonumber
\begin{split}
\sum_{j=-L}^L \Ex{Z_{n,j}^2 1\Par{\abs{Z_{n,j}}>\epsilon}}&=\Ex{n \Par{Z_{n,j}}^2 1\Par{\abs{\sqrt n Z_{n,j}}>\epsilon\sqrt n}}\\
&\leq\Ex{n \Par{Z_{n,j}}^2 1\Par{\abs{\sqrt n Z_{n,j}}>M}}\\
&=\Var{\sqrt{n}Z_{n,j}1(\abs{\sqrt n Z_{n,j}}>M)}\underset{n \rightarrow \infty}{\longrightarrow}\Var{\Phi_M}
\end{split}
\end{equation}
(due to (\ref{small_conv})).

It remains to show that there exists a summable sequence $(\theta_r)_{r\in \N}$ so that the upper bounds for covariances required in Neumann's Theorem hold (equations (\ref{cov_req_1}) and (\ref{cov_req_2}), for all relevant cases). From the $m$-dependence of the finite sequence $\Set{Z_{n,j}}_{j\in\Lambda_L}$, we deduce that the left hand sides of (\ref{cov_req_1}) and (\ref{cov_req_2}) equal $0$ whenever $r>m$, so we conclude by finding some $\theta_1,\ldots,\theta_m<\infty$. A straightforward computation shows that (\ref{cov_req_1}) holds as long as $\theta_r\geq 1$.
\comment{
Note that we always have
$$\Var{g\Par{Z_{n,s_1},\ldots,Z_{n,s_u}}Z_{n,s_u}}\leq\Ex{g\Par{Z_{n,s_1},\ldots,Z_{n,s_u}}^2 Z_{n,s_u}^2}\leq \Ex{Z_{n,s_u}^2}$$
(since $||g||_{\infty}\le 1$), therefore
\begin{equation} \nonumber
\begin{split}
& \abs{\Cov{g\Par{Z_{n,s_1},...,Z_{n,s_u}}Z_{n,s_u}\ ,\ Z_{n,t_1}}}\leq \\
& \sqrt{\Var{g\Par{Z_{n,s_1},...,Z_{n,s_u}}Z_{n,s_u}}\Var{Z_{n,t_1}}} \leq \sqrt{\Ex{Z_{n,s_u}^2}\Ex{Z_{n,t_1}^2}} =\Ex{Z_{n,s_u}^2}\\
\end{split}
\end{equation}
since the distribution of $Z_{n,j}$ is independent of $j$. This tells us that (\ref{cov_req_1}) holds as long as $\theta_1,\ldots,\theta_m\geq 1$.}
To prove (\ref{cov_req_2}), we use
$$\Var{g\Par{Z_{n,s_1},\ldots,Z_{n,s_u}}}\leq\Ex{g\Par{Z_{n,s_1},\ldots,Z_{n,s_u}}^2}\leq 1$$
(as $||g||_{\infty}\le 1$) to obtain
\begin{equation} \nonumber
\begin{split}
& \abs{\Cov{g\Par{Z_{n,s_1},...,Z_{n,s_u}}\ ,\ Z_{n,t_1}Z_{n,t_2}}}\leq \\
& \sqrt{\Var{g\Par{Z_{n,s_1},...,Z_{n,s_u}}}\Var{Z_{n,t_1}Z_{n,t_2}}} \leq \sqrt{\Var{Z_{n,t_1}Z_{n,t_2}}},
\end{split}
\end{equation}
and we conclude by showing that for some $\theta<\infty$,
$$\sqrt{\Var{Z_{n,t_1}Z_{n,t_2}}}\leq\frac 1n \theta$$
holds for every $n=2L+1$ and $t_1,t_2\in\Lambda_L$. Equivalently, we will show that
$$\sup_{n,t_1,t_2}\Var{\sqrt n Z_{n,t_1}\cdot \sqrt n Z_{n,t_2}}<\infty.$$
From (\ref{ind_hyp}) we deduce that
$$\sup_n\Var{\sqrt n Z_{n,t_1}\cdot \sqrt n Z_{n,t_2}}<\infty$$
for every $t_1,t_2\in\Z$. Furthermore, since our initial variables $\Set{Y_i}_{i\in\Zd}$ are identically distributed, the value of $\Var{\sqrt n Z_{n,t_1}\cdot \sqrt n Z_{n,t_2}}$ depends only on $n$ and $t_2-t_1$, and since our variables are $m$-dependent, it is enough to consider $\abs{t_2-t_1}\in\Set{0,1,\ldots,m+1}$. This concludes our proof.
\end{proof}

\begin{proof}[Proof of Theorem \ref{main_CLT}]
Theorem \ref{main_CLT} will follow from Proposition \ref{prop_aux_CLT}, applied to the variables
$$Y_i=\sum_{\beta\in B}a_\beta\Par{X^{\beta^i} - \Ex{X^{\beta^i}}}.$$
Clearly the variables $\Set{Y_i}_{i\in\Zd}$ are identically distributed (since $\Set{X_n}_{n\in\Zd}$ are), and $\Ex{Y_i}=0$. Since every $X_n$ has finite moments, so do $Y_i$ (as a finite sum of products of the variables $\Set{X_n}_{n\in\Zd}$). In particular, $\Ex{\abs{Y_i}^3}<\infty$.\\
Since $a_\beta\neq 0$ only for finitely many $\beta\in B$, one can find sufficiently large $m$, such that whenever $\abs{j-i}>m$ and $a_\beta,a_\gamma\neq 0$, the supports of $\beta^i$ and $\gamma^j$ are disjoint. From here it follows that $\Set{Y_i}_{i\in \Zd}$ is $m$-dependent.
\end{proof}


\end{document}